\documentclass[envcountsame,a4paper,10pt]{llncs}
\usepackage[margin=3.5cm]{geometry}
\def\Vhrulefill{\leavevmode\leaders\hrule height 0.7ex depth \dimexpr0.4pt-0.7ex\hfill\kern0pt}

\usepackage{times}

\usepackage{enumerate}
\usepackage{amsmath}
\usepackage{amssymb}
\usepackage{amsfonts}
\usepackage{ifthen}
\usepackage{xspace}
\usepackage{xcolor}
\usepackage{graphicx}
\usepackage{wrapfig}
\usepackage{comment}
\usepackage{verbatim}

\newcommand{\logicClFont}[1]{\mathsf{#1}}        
\newcommand{\complexityClassFont}[1]{\textsl{#1}} 
\newcommand{\mathCommandFont}[1]{\mathrm{#1}}     
\newcommand{\problemFont}[1]{\textsc{#1}}     

\newcommand{\literal}[1]{{\protect\ensuremath{#1}}\xspace}
\renewcommand{\problem}[1]{\literal{\problemFont{#1}}}
\newcommand{\logic}[1]{\literal{\logicClFont{#1}}}
\newcommand{\commandOperator}[2]{\literal{\mathord{\mathCommandFont{#1}\ifthenelse{\equal{#2}{}}{}{(\nobreak#2\nobreak)}}}}

\newcommand{\todo}[2][]{\textnormal{\color{red}\scriptsize+++\ifthenelse{\equal{#1}{}}{}{#1 says: }#2+++}}

\newcommand{\NEXPTIME}{\literal{\complexityClassFont{NEXPTIME}}}

\newcommand{\true}{\protect\raisebox{-1pt}{$\true$}}
\newcommand{\false}{\protect\ensuremath\raisebox{-1pt}{$\false$}}
\newcommand{\dep}[1][]{\commandOperator{=}{#1}}
\newcommand{\fr}[1][]{\commandOperator{Fr}{#1}}

\newcommand{\dom}[1][]{\commandOperator{dom}{#1}}
\newcommand{\rng}[1][]{\commandOperator{range}{#1}}
\newcommand{\rel}[1][]{\commandOperator{rel}{#1}}
\newcommand{\df}{\logic{D}}

\newcommand{\ind}{\logic{Ind}}
\newcommand{\exc}{\logic{Exc}}
\newcommand{\inc}{\logic{Inc}}
\newcommand{\incexc}{\logic{Inc/Exc}}
\newcommand{\dtwo}{\literal{\df^2}}

\newcommand{\fo}{\logic{FO}}
\newcommand{\fotwo}{\literal{\fo^2}}
\newcommand{\foc}{\logic{FOC}}
\newcommand{\foctwo}{\logic{FOC^2}}
\newcommand{\mA}{\literal{\mathfrak{A}}}
\newcommand{\mB}{\literal{\mathfrak{B}}}
\newcommand{\mG}{\literal{\mathfrak{G}}}

\newcommand{\mf}{\mathfrak}
\renewcommand{\true}{\literal{\top}}
\renewcommand{\false}{\literal{\bot}}
\newcommand{\sat}[1][]{\commandOperator{\problem{Sat}}{#1}}

\newcommand{\finsat}[1][]{\commandOperator{\problem{FinSat}}{#1}}

\newcommand{\mi}{\mathit}
\newcommand{\dfn}{:=}
\newcommand{\cC}{\mathcal{C}}

\title{Decidability of predicate logics with team semantics
\thanks{Juha Kontinen was supported by grant 292767 of the
Academy of Finland. Antti Kuusisto was supported by the ERC grant 647289
``CODA" and the Jenny and Antti Wihuri Foundation.
Jonni Virtema was supported by grant 266260 of the
Academy of Finland and a grant by the Finnish Academy of Science and Letters.}}

\author{Juha Kontinen\inst{1} \and Antti Kuusisto\inst{2} \and Jonni~Virtema\inst{1,3}}

\institute{
University of Helsinki, Finland\\
\and
University of Bremen, Germany\\
\and
Leibniz Universit{\"a}t Hannover, Germany\\
\email{\{juha.kontinen, jonni.virtema\}@helsinki.fi}\\
\email{kuusisto@uni-bremen.de}
}

\begin{document}
\maketitle

\begin{abstract}
We study the complexity of predicate logics based on team semantics.
We show that the satisfiability problems of 
two-variable independence logic and inclusion logic 
are both $\NEXPTIME$-complete. Furthermore,
we show that the validity problem of
two-variable dependence logic is undecidable,
thereby solving an open problem from the team
semantics literature.
We also briefly analyse the 
complexity of the Bernays-Sch\"onfinkel-Ramsey prefix classes of
dependence logic.
\end{abstract}

\section{Introduction}

%
%
%
The satisfiability problem of 
\emph{two-variable logic} $\fo^2$ was shown to
be $\NEXPTIME$-complete in \cite{grkova97}.
The extension of two-variable logic with counting quantifiers, $\foctwo$,
was proved decidable in \cite{grotro97,paszte97}, and it
was subsequently shown to be
$\NEXPTIME$-complete in \cite{IEEEonedimensional:pratth}.
Research on extensions and variants of
two-variable logic is currently very active.
Recent research efforts have mainly concerned decidability
and complexity issues
in restriction to particular classes of 
structures and also questions
related to different built-in features and operators that
increase the expressivity of the base language.
Recent articles in the field include for example
\cite{IEEEonedimensional:kieronski},
\cite{IEEEonedimensional:charatonic},
\cite{helkuu14},
\cite{KMP-HT14},
%
%
\cite{IEEEonedimensional:tendera},
and several others.
%

%
%
%
%
%

%
In this article we study two-variable fragments of logics based on \emph{team semantics}.
Team semantics was originally conceived in
\cite{ho97} in the context of \emph{independence friendly} (IF) \emph{logic} \cite{hisa89}.
In \cite{va07}, V\"{a}\"{a}n\"{a}nen introduced \emph{dependence logic},
which is a novel approach to IF logic based on new atomic formulas $\dep[x_1,...x_k,y]$
stating that the interpretation of the variable $y$ is functionally determined by
the interpretations of the variables $x_1,...,x_k$.
After the introduction of dependence logic, research on logics based on team semantics
has been active. Several different logics with different applications have
been suggested.
%
%
In particular, team semantics has proved to be a powerful framework
for studying different kinds of \emph{dependency notions}.
\emph{Independence logic} \cite{independencelogic} is a
variant of dependence logic that extends first-order logic by
new atomic formulas $x_1,...,x_k\, \bot\, y_1,...,y_l$
with the intuitive meaning that the interpretations of
the variables $x_1,...,x_k$ are informationally independent of
the interpretations of the variables $y_1,...,y_l$.
\emph{Inclusion logic} \cite{inclusionlogic} extends first-order logic by
atomic formulas $x_1,...,x_k\, \subseteq\, y_1,...,y_k$,
whose intuitive meaning is that tuples interpreting
the variables $x_1,...,x_k$ are also
tuples interpreting $y_1,...,y_k$.
Currently dependence, independence and inclusion logics are 
the three most important and most widely studied systems based on team semantics.
Both dependence logic and
independence logic are equiexpressive with existential
second-order logic (see \cite{va07}, \cite{independencelogic}),
and thereby capture $\mathrm{NP}$.
Curiously, inclusion logic is equiexpressive with \emph{greatest fixed point logic} (see \cite{GH2013}),
and thereby characterizes $\mathrm{P}$ on finite ordered
models. While the descriptive complexity of most known logics based on
team semantics is understood reasonably well, 
the complexity of related satisfiability problems has
received somewhat less attention.
The satisfiability problem of the
two-variable fragment of dependence logic and $\mathrm{IF}$-logic have been studied in \cite{kokulovi2}.
It is shown that 
while the two-variable $\mathrm{IF}$-logic is undecidable,
the corresponding fragment of dependence logic is $\NEXPTIME$-complete.
In this article we establish that the satisfiablity problems of the
two-variable fragments of
independence and inclusion logics are
likewise $\NEXPTIME$-complete. This result is established via proving a
more general theorem that implies also a range of other decidability results for a
variety of team-semantics-based logics with generalized dependency notions.
Furthermore, we prove that the \emph{validity} problem of
two-variable dependence logic is undecidable; this result is the 
main result of the paper. The problem 
has been open for some time in the team semantics literature and
has been explicitly posed in, e.g., \cite{dukovo16}, \cite{kokulovi2},
\cite{jonnithesis},
and elsewhere.
%

%
%
%
%
In addition to studying two-variable logics, we study the Bernays-Sch\"onfinkel-Ramsey prefix class, i.e.,
sentences with the quantifier prefix $\exists^*\forall^*$.
%
%
%
We show that---as in the case of
ordinary first-order logic---the prefix class $\exists^*\forall^*$ of
$\fo(\mathcal{A})$ is decidable for any uniformly polynomial time computable class $\mathcal{A}$ of
generalized dependencies closed under substructures.
We prove inclusion in $2\NEXPTIME$, and furthermore, for 
vocabularies of fixed arity, we show $\NEXPTIME$-completeness.
We also prove a partial converse of
the result concerning logics $\fo(\mathcal{A})$
with a decidable prefix class $\exists^*\forall^*$,
see Theorem \ref{whatevertheorem}.
%

%
%
%
%
%

%
%

\section{Preliminaries}\label{preliminaries}
%

The domain of a structure $\mA$ is denoted by $A$. We assume that the
reader is familiar with first-order logic \fo. The extension of \fo with counting quantifiers $\exists ^{\ge i}$ is denoted by \foc. 
The two-variable fragments $\fo^2$ and $\foc^2$ are the fragments of \fo and \foc 
with formulas in which only the variables $x$ and $y$ appear.
We let $\Sigma_1^1$ denote the fragment of formulas of second-order logic of the form
$\exists X_1...\exists X_k\, \varphi$, where $X_1,...,X_k$ are relation symbols 
and $\varphi$ a first-order formula. $\Sigma_1^1(\foctwo)$ is 
the extension of $\foctwo$ consisting of formulas of the
form $\exists X_1...\exists X_k\, \chi$, where $X_1,...,X_k$ are relation symbols 
and $\chi$ a formula of $\foctwo$. 
%

%
%
%



\subsection{Logics based on team semantics}
Let $\mathbb{Z}_+$ denote the set of positive integers,
and let $\mathrm{VAR} = \{\, v_i\ |\ i\in\mathbb{Z}_+\ \}$ be
the set of  exactly all first-order \emph{variable symbols}.
We mainly use metavariables $x,y,z,x_1,x_2$, etc.,
in order to refer to variable symbols in $\mathrm{VAR}$.
We let $\overline{x},\overline{y},\overline{z},\overline{x}_1,\overline{x}_2$, etc.,
denote finite nonempty tuples of variable symbols, i.e., 
tuples in $\mathrm{VAR}^n$ for some $n\in\mathbb{Z}_+$.
When we study two-variable logics, we use the metavariables $x$ and $y$,
and assume they denote distinct variables in $\mathrm{VAR}$.
Let $D\subseteq\mathrm{VAR}$ be a \emph{finite}, possibly empty set.
Let $\mathfrak{A}$ be a model.
We do not allow for models to have an empty domain,
so $A\not=\emptyset$.
A function $s:D\rightarrow A$ is called an
\emph{assignment} with codomain $A$.
If $\overline{x}=(x_1,\dots,x_n)$, we denote $(s(x_1),\dots,s(x_n))$ by $s(\overline{x})$.
We let $s[a/x]$ denote the variable assignment
with the domain $D\cup\{\ x\ \}$ and 
codomain $A$ defined such that 
$s[a/x](y) = a$ if $y = x$,
and
$s[a/x](y) = s(y)$ if $y \not= x$.
Let $T\in\mathcal{P}(A)$, where $\mathcal{P}$
denotes the power set operator.
We define 
$s[\,  T/x\, ]\ =\ \{\ s[a/x]\ |\ a\in T\ \}.$
%
%
%
%

%
Let $D\subseteq\mathrm{VAR}$ be a finite, possibly empty set of
first-order variable symbols. Let $X$ be a set of assignments
$s:D\rightarrow A$. Such a set $X$ is a \emph{team}
with the \emph{domain} $D$ and \emph{codomain} $A$.
Note that the empty set is a team, as is the set $\{\emptyset\}$
containing only the empty assignment.
The team $\emptyset$ does not have a unique
domain; any finite subset of $\mathrm{VAR}$ is a
domain of $\emptyset$.
The domain of the team $\{\emptyset\}$ is $\emptyset$.
Let $X$ be a team with the domain $D$ and codomain $A$.
Let $T\subseteq A$. We define
$X[ T/x] = \{\ s[a/x]\ |\ a\in T,\ s\in X\ \}$.
%
%
%
Let $F:X\rightarrow \mathcal{P}(A)$ be a function.
We define
$X[\, F/x\, ]=\bigcup\limits_{s\, \in\, X}s[\, F(s)/x\, ]$.
Let $C\subseteq A$. We define 
$X\upharpoonright C = \{\, s\in X\ |\ s(x)\in C\text{ for all }x\in D\ \}.$
Let $X$ be a team with domain $D$.
Let $k\in\mathbb{Z}_+$, and 
let $y_1,...,y_k$ be variable symbols.
Assume that $\{y_1,...,y_k\} \subseteq D$. We define
$\mathrm{rel}\bigl(X,(y_1,...,y_k)\bigr)\ =\ \{\ \bigl(s(y_1),...,s(y_k)\bigr)\ |\ s\in X\ \}.$
%
%
%
%

%
Let $\tau$ be a relational vocabulary, i.e., a vocabulary containing
relation symbols only. (In this article we consider only relational vocabularies.)
The syntax of a logic based on team semantics is usually given
in negation normal form. We shall also follow this convention in the current article.
For this reason, we define the syntax of first-order logic as follows.
\[
\varphi\ \ ::=\ \ R(\overline{x})\ |\ \neg R(\overline{x})\ |\
x_1 = x_2\ |\ \neg x_1 = x_2\ |\  (\varphi_1\vee\varphi_2)\ |\ (\varphi_1\wedge\varphi_2)\ |\
\ \exists x\varphi\ |\ \forall x \varphi\ ,
\]
%
%
where $R\in\tau$. The first four formula formation 
rules above introduce \emph{first-order literals} to the language. Below we shall consider
logics $\fo(\mathcal{A})$, where the above syntax is
extended by clauses of the type $A_Q\,(\overline{y}_1,...,\overline{y}_k)$.
Here $A_Q$ is (a symbol corresponding to) a \emph{generalized atom}
in $\mathcal{A}$ and each $\overline{y}_i$ is a
tuple of variables. Before considering such novel atoms, let us define \emph{lax team semantics}
for first-order logic.
\begin{definition}[\cite{ho97,va07}]\label{def:semantics}
Let $\mA$ be a model and $X$ a team with codomain $A$. The satisfaction relation
$\mA\models _X \varphi$ is defined as follows.
\begin{enumerate}
\item If $\varphi$ is a first-order literal, then $\mA\models_X \varphi$ iff for all $s\in X$: $\mA,s\models_{\fo}\varphi$.
Here $\models_{\fo}$ refers to the ordinary Tarskian satisfaction relation of first-order logic.

\item $\mA\models_X \psi \wedge \varphi$ iff $\mA\models _X \psi$ and $\mA\models _X \varphi$.
\item $\mA\models_X \psi \vee \varphi$ iff there exist teams $Y$ and $Z$ such that $X=Y\cup Z$, $\mA\models_Y \psi$, and $\mA\models _Z \varphi$.

\item   $\mA \models_X \exists x\, \psi$ iff $\mA \models _{X[F/x]} \psi$ for
some $F\colon X\to (\mathcal{P}(A)\setminus\{\emptyset\})$.

\item $\mA \models_X \forall x\, \psi$ iff $\mA \models _{X[A/x]} \psi$.
\end{enumerate}
\end{definition}
%
%
%
%
%
%
%
%
%
%
%
%
%
%
%
%
%
Finally, a sentence $\varphi$ is true in a model $\mA$  ($\mA\models \varphi$)  if $\mA\models _{\{\emptyset\}} \varphi$. 

%
%
%
\begin{proposition}[\cite{ho97,va07}]\label{whatever}
Let $\psi$ be a formula of first-order logic.
We have $\mA \models_X \psi$ iff\, 
$\mA,s \models_{\fo} \psi$ for all $s\in X$.
\end{proposition}
%

%
%

%
In this paper we consider first-order logic extended with generalized dependency atoms. Before
formally introducing the notion of a generalized dependency atom,
we recall some particular atoms familiar from the literature related to team semantics.
\emph{Dependence atoms} $\dep[x_1,\dots,x_n,y]$, inspired by the slashed quantifiers of Hintikka and Sandu \cite{hisa89}, were introduced by V\"a\"an\"anen \cite{va07}. The intuitive meaning of the atom $\dep[x_1,\dots,x_n,y]$ is that the value of the variable $y$ depends solely on the values of the variables $x_1,\dots,x_n$. The semantics for dependence atoms is defined as follows:
\[
\mA\models_X \dep[x_1,\dots,x_n,y]  \text{ iff } \forall s,s'\in X: \text{ if } s((x_1,\dots,x_n))=s'((x_1,\dots,x_n)) \text{ then } s(y)=s'(y).
\]
\emph{Dependence logic} (\df) is the extension of first-order logic with dependence atoms.

While dependence atoms of dependence logic declare dependences between variables,
\emph{independence atoms}, introduced by Gr\"adel and V\"a\"an\"anen \cite{independencelogic}, 
do just the opposite; independence atoms are used to declare independencies between variables.
Independence atom is an atomic formula of the form
$(x_1,...,x_k)\, \bot_{(z_1,\dots,z_t)}\, (y_1,...,y_l)$
with the intuitive meaning that for any fixed interpretation of the variables $z_1,\dots,z_t$, the interpretations of
the variables $x_1,...,x_k$ are independent of
the interpretations of the variables $y_1,...,y_l$.
The semantics for independence atoms is defined as follows:
\begin{align*}
&\mA\models_X (x_1,...,x_k)\, \bot_{(z_1,\dots,z_t)}\, (y_1,...,y_l) \,\text{ iff } \, \forall s,s'\in X\, \exists s''\in X: \bigwedge_{i\leq t} s(z_i)=s'(z_i) \\
&\text{implies that }\bigwedge_{i\leq k} s''(x_i)=s(x_i) \wedge \bigwedge_{i\leq t} s''(z_i)=s(z_i) \wedge \bigwedge_{i\leq l} s''(y_i)=s'(y_i).
\end{align*}
\emph{Independence logic} (\ind) is the extension of first-order logic with independence atoms.
Galliani \cite{inclusionlogic} introduced \emph{inclusion} and
\emph{exclusion atoms}. The intuitive meaning of the inclusion atom $(x_1,\dots,x_n)\subseteq (y_1,\dots,y_n)$ is that tuples interpreting the variables $x_1\dots,x_n$ are also tuples interpreting $y_1,\dots,y_n$. The intuitive meaning of the exclusion atom $(x_1,\dots,x_n)\mid (y_1,\dots,y_n)$ on the other hand is that tuples interpreting the variables $x_1\dots,x_n$ and the tuples interpreting $y_1,\dots,y_n$ are distinct. The semantics for inclusion atoms and exclusion atoms is defined as follows:
\begin{align*}
\mA\models_X (x_1,\dots,x_n)\subseteq (y_1,\dots,y_n) & \text{ iff } \forall s\in X \,\exists s'\in X: s((x_1,\dots,x_n))=s'((y_1,\dots,y_n)),\\
\mA\models_X (x_1,\dots,x_n)\mid (y_1,\dots,y_n) & \text{ iff } \forall s,s'\in X:  s((x_1,\dots,x_n))\not=s'((y_1,\dots,y_n)).
\end{align*}
The extension of first-order logic with inclusion atoms (exclusion atoms) is called \emph{inclusion logic} (\emph{exclusion logic}) and denoted by $\inc$ ($\exc$). The extension of first-order logic with both inclusion atoms and exclusion atoms is called \emph{inclusion/exclusion logic} and denoted by $\incexc$.

\subsection{Generalized atoms}
In this section we first give the well known definition of generalized quantifiers (Lindstr\"om quantifiers \cite{Lindstrom66}). We then show how each generalized quantifier naturally gives rise to a generalized atom. Finally, we discuss on some fundamental properties of first-order logic extended with generalized atoms.
Generalized atoms were first defined in \cite{kuusigen}.

Let $(i_1,...,i_n)$ be a nonempty sequence of positive integers.
A generalized quantifier of the type $(i_1,...,i_n)$ is a class $\mathcal{C}$ of structures
$(A,B_1,...,B_n)$ such that the following conditions hold.
\begin{enumerate}
\item
$A\not=\emptyset$, and for each $j\in\{1,...,n\}$, we have $B_j \subseteq A^{i_j}$.
\item
If $(A',B_1',...,B_n')\in\mathcal{C}$ and if there is an isomorphism $f:A'\rightarrow A''$
from $(A',B_1',...,B_n')$ to another structure $(A'',B_1'',...,B_n'')$,
then $(A'',B_1'',...,B_n'')\in\mathcal{C}$.
\end{enumerate} 
Let $Q$ be a generalized quantifier of the type $(i_1,...,i_n)$.
Let $\mathfrak{A}$ be a model 
with the domain $A$.
We define $Q^{\mathfrak{A}}$ to be the set
$\{\ (B_1,...,B_n)\ |\ (A,B_1,...,B_n)\in Q\ \}.$
Let $n$ be a positive integer.
Let $Q$ be a generalized quantifier of the type $(i_1,...,i_n)$.
Extend the syntax of first-order logic
with atomic expressions of the type
$A_{Q}(\overline{y}_1,...,\overline{y}_n),$
where each $\overline{y}_j$ is a tuple of variables of length $i_j$.
Let $X$ be a team whose domain contains all variables
occurring in the tuples $\overline{y}_1,...,\overline{y}_n$.
Extend team semantics such that
$\mathfrak{A}\models_X A_{Q}(\overline{y}_1,...,
\overline{y}_n)$
if and only if
$\bigl(\mathrm{rel}(X,\overline{y}_1),...,
\mathrm{rel}(X,\overline{y}_n)\bigr)\in Q^{\mathfrak{A}}.$
The generalized quantifier $Q$ defines a
\emph{generalized atom} $A_{Q}$
of the type
$(i_1,...,i_n)$.
A generalized atom $A_Q$ is \emph{downwards closed} if
for all $\mathfrak{A}$, $X$ and $\overline{y}_1,...,\overline{y}_k$,
it holds that if
$\mathfrak{A}\models_X A_Q(\overline{y}_1,...,\overline{y}_k)$ and $Y\subseteq X$,
then $\mathfrak{A}\models_Y A_Q(\overline{y}_1,...,\overline{y}_k)$.
Similarly, a generalized atom $A_Q$ is \emph{closed under substructures} if for all $\mathfrak{A}$, $X$ and $\overline{y}_1,...,\overline{y}_k$,
it holds that if
$\mathfrak{A}\models_X A_Q(\overline{y}_1,...,\overline{y}_k)$, $\mathfrak{A}':=\mathfrak{A}\upharpoonright B$ and $X':=X\upharpoonright B$ for some  $B\subseteq A$,
then we have $\mathfrak{A}'\models_{X'} A_Q(\overline{y}_1,...,\overline{y}_k)$.
Finally,  a generalized atom $A_Q$ is \emph{universe independent} if for all $\mA$, $\mB$, $X$ and
$\overline{y}_1,...,\overline{y}_k$,  where both $A$ and $B$
are codomains for $X$, it holds that
$\mathfrak{A}\models_X A_Q(\overline{y}_1,...,\overline{y}_k)$
if and only if $\mathfrak{B}\models_X A_Q(\overline{y}_1,...,\overline{y}_k)$.
Let $\varphi$ be a formula of first-order logic, possibly extended with generalized atoms.
The set $\mathrm{Fr}(\varphi)$ of \emph{free variables} of $\varphi$ is defined
in the same way as in first-order logic. The set $\mathrm{Fr}(A_Q(\overline{y}_1,...,\overline{y}_k))$
of course contains exactly all variable that occur in the tuples $\overline{y}_i$.
The satisfiability problem of a (possibly team-semantics-based) logic $L$ takes as an
input a sentence of $L$ and asks whether $\mathfrak{A}\models\varphi$
for some model $\mathfrak{A}$. The validity problem asks,
given a sentence $\varphi$, whether $\mathfrak{A}\models\varphi$
for all models $\mathfrak{A}$.
Let $k\in\mathbb{Z}_+$ and let $A_Q$ be a generalized atom of the type $(i_1,...,i_n)$,
where $i_j\leq k$ for each $j$.
%
%
Let $\varphi(R_1,...,R_n)$ be a sentence of $\Sigma_1^1(\foc^k)$ with
unquantified relation symbols $R_1,...,R_n$ of
arities $i_1,...,i_n$, respectively.
Assume that for all models $\mathfrak{A}$ and teams $X$
with codomain $A$ and domain containing the
variables in $A_Q(\overline{x}_1,...,\overline{x}_n)$, we have 
$\mathfrak{A}\models_X A_Q(\overline{x}_1,...,\overline{x}_n)$ iff
$$\bigl(\mathfrak{A}, R_1 := \mathrm{rel}(X,\overline{x}_1),...,
R_n := \mathrm{rel}(X,\overline{x}_n)\bigr)\models_{\fo}\varphi(R_1,...,R_n).$$
Then we say that the atom
$A_Q$ is definable in $\Sigma_1^1(\foc^k)$.
%


We now show that, \emph{for any generalized atom} $A_Q$, the logic $\fo(A_Q)$ has the so-called locality property. We also show that, for a downwards closed atom  $A_Q$, all formulas of $\fo(A_Q)$ satisfy the
downwards closure property. These two properties  have previously turned out to be very useful in the study of  dependence logic.
%

%
%
%
%
%

%
Let $X$ be a team with domain $\{x_1,\ldots,x_k\}$,
and let $V\subseteq \{x_1,\ldots,x_k\}$. We denote by
$X(V)$ the team $\{s\upharpoonright V \mid s\in X\}$ with the domain $V$. The following proposition
 shows that the truth of an $\fo(A_Q)$-formula  depends only on the interpretations of the variables occurring free in the formula.
The proof uses the fact that generalized atoms satisfy the claim by definition. Otherwise the proof is identical to the corresponding proof given in \cite{inclusionlogic}.
\begin{proposition}[Locality]\label{freevar} Let $A_Q$ be a generalized atom and  $\varphi\in \fo(A_Q)$ a formula. If 
$V\supseteq \fr(\varphi)$, then $\mA \models _X\varphi$ if and only if $\mA \models _{X(V)} \varphi$.
\end{proposition}

The next proposition is also very useful.
%
The proof is almost identical to the corresponding proof for dependence logic, see \cite{va07}. The additional case for generalized atoms follows by the assumption of downwards closure.
\begin{proposition}[Downward closure]\label{Downward closure}
Let $A_Q$ be a downwards closed generalized atom.  Suppose $\varphi$ is an $\fo(A_Q)$-formula,  
$\mA$  a model, and $Y\subseteq X$ teams. Then $\mA\models_X \varphi$ implies $\mA\models_Y\varphi$.
\end{proposition}



%
\section{Satisfiability problems of logics $\fotwo(\mathcal{A})$}\label{sec:satfo2}
%
%
In this section we show that for any finite collection 
$\mathcal{A}$ of\, $\Sigma_1^1(\foctwo)$-definable atoms $A_Q$,
 both $\sat[\fotwo(\mathcal{A})]$ and
 $\finsat[\fotwo(\mathcal{A})]$ are \NEXPTIME-complete. Our proof relies on a translation
from  $\fotwo(\mathcal{A})$ into $\Sigma_1^1(\foctwo)$ and the fact that \sat[\foctwo]
and \finsat[\foctwo] are \NEXPTIME-complete \cite{IEEEonedimensional:pratth}.
We start by establishing a more general translation.
We show that for every $k\geq 1$ and every $\Sigma_1^1(\foc^k)$ definable atom $A_Q$,
we have $\fo^k(A_Q)\leq \Sigma_1^1(\foc^k)$. Note that strictly speaking $\fo^k(A_Q)$
uses only one atom $A_Q$ instead of a finite collection $\mathcal{A}$ of atoms, but our proof
below generalizes \emph{directly} to the case with a finite collection of atoms.
The reason for considering a single atom is simply to keep the notation light.
When considering $k$-variable logic, we let $\{x_1,...,x_k\}$ denote the $k$ distinct
variables used in the syntax of the logic, and we let $\mathit{rel}(X)$ denote $\mathit{rel}\bigl(X,(x_1,...,x_k)\bigr)$.
%
%
%
%
%
%
%
%
%
The following lemma is possibly the technically most 
involved part of our argument in this section for establishing
decidability of two-variable inclusion and independence logics.
The proof significantly modifies and
extends the argument establishing
Lemma~3.3.14 of \cite{jonnithesis}. See also \cite{kokulovi2}
and Theorem 6.2 in \cite{va07}.
\begin{lemma}\label{DtoESO}
Assume that $k,t\geq 1$. Let $\tau$ be a relational vocabulary, let $R\not\in\tau$ be a $k$-ary relation symbol and let $A_Q$ be a $\Sigma_1^1(\foc^k)$-definable atom of type $(i_1,\dots,i_t)$, where $i_j\leq k$ for each $j$. For every formula $\varphi\in \fo^k(A_Q)$ there exists a sentence 
$
\varphi^*\in \Sigma_1^1(\foc^k)(\tau\cup\{R\})
$
such that
for every model $\mA$ and team $X$ with codomain $A$ and $\dom(X)=\{x_1,\dots,x_k\}$, we have
\begin{equation}\label{dtoesoequiv}
\mA\models_X \varphi \quad\text{ iff }\quad \big(\mA,\rel[X]\big)\models \varphi^*,
\end{equation}
where $(\mA,\rel[X])$ is the expansion $\mA'$ of\, $\mA$ into the vocabulary $\tau\cup\{R\}$ such that
$R^{\mA'} := \rel[X]$. Moreover $\varphi^*$ is computable from $\varphi$ in polynomial time.
%
%
%
\end{lemma}
\begin{proof}
Fix $k\geq 1$ and the $\Sigma_1^1(\foc^k)$-definable atom
$A_Q$. Let $(i_1,\dots,i_t)$, where $i_j\leq k$ for each $j$, be the type of $A_Q$.
Let $\varphi_{A_Q}(R_1,\dots,R_t)$ be the $\Sigma_1^1(\foc^k)$-sentence that
defines $A_Q$. We will define a translation
\[
tr_k:\fo^k(A_Q)(\tau)\ \rightarrow\ \Sigma_1^1(\foc^k)(\tau\cup\{R\})
\]
inductively. Below we always assume that the quantified relations $S$ and $T$ are fresh, i.e., they are assumed not to appear in
$tr_k(\psi)$ or $tr_k(\vartheta)$. Notice that for every $\fo^k(A_Q)$-formula $\varphi$,
we have $tr_k(\varphi)=\exists S_1\dots \exists S_n\varphi'$
for some $k$-ary relation variables $S_1\dots S_n$ ($n\in\mathbb{N}$) and some
$\foc^k$-formula $\varphi'$. The translation
$tr_k$ is defined as follows.
\begin{enumerate}
\item If $\varphi$ is a first-order literal (and thus not a generalized atom), then
\[tr_k(\varphi) :=
\forall x_1\dots\forall x_k\big( R(x_1,\dots,x_k)\rightarrow \varphi\, \big).
\]
%
%
\item Assume that $\varphi$ is a
generalized atom $A_Q(\overline{y}_1,\dots,\overline{y}_t)$,
where $\overline{y}_j\in \{x_1,\dots,x_k\}^{i_j}$ for each $j\leq t$.
Let $\overline{Y}$ and $\psi\in \foc^k(R_1,...,R_t)$ be such that
\(
 \varphi_{A_Q} = \overline{\exists Y}\psi.
\)
For technical reasons, we will simulate $i_j$-ary relations by $k$-ary relations.
Define that, for each $j\leq t$,
$$\mathrm{Id}_j:=\{(l,m)\in \mathbb{N}^2 \mid y_{j_l} \text{ and } y_{j_m} \text{ denote the same variable symbol} \},$$
where $y_{j_l}$ ($y_{j_m}$) denotes the $l$-th ($m$-th) element of $\overline{y}_j$.
Now $tr_k(\varphi)$ is defined to be the formula
\begin{align*}
\overline{\exists Y}\,\exists T_1\dots\exists T_t  \Big(\bigwedge_{j\leq t} \Big(\varphi_{j\text{-}\mathit{padding}} \, \wedge \, \varphi_{j\text{-}\mathit{identities}}\Big)
\, \wedge\, \psi'\, \Big),
\end{align*}
where the relation variables $T_j$ and formulas $\varphi_{j\text{-}\mathit{padding}}$, $\varphi_{j\text{-}\mathit{identities}}$ and $\psi'$ are defined as follows.
Each variable $T_j$ is a fresh \emph{$k$-ary} relation variable. 
The formula $\psi'$ is the conjunction $\psi''\wedge\bigwedge_{j\leq t}\chi_j$,
where $\psi''$ and $\chi_j$ are as follows.
The conjunct $\psi''$ is obtained from $\psi$ by replacing each atomic formula $R_j(z_1,\dots,z_{i_j})$ by
$T_j(z_1,\dots,z_{i_j},z_1,\dots,z_1)$.
For each $j\leq t$, $\chi_j$ is the formula
\begin{multline*}
\forall x_1\dots\forall x_k\Big(
\exists x_{i_{j+1}}\, T_j(x_1,...,x_{i_j},x_{i_{j+1}},...,x_{i_{j+1}})\\ \rightarrow \forall x_{i_{j+1}}\, T_j(x_1,...,x_{i_j},x_{i_{j+1}},...,x_{i_{j+1}})\Big),
\end{multline*}
where in the
case $i_j = k$ the formulas $\exists x_{i_{j+1}}\, T_j(x_1,...,x_{i_j},x_{i_{j+1}},...,x_{i_{j+1}})$
and $$\forall x_{i_{j+1}}\, T_j(x_1,...,x_{i_j},x_{i_{j+1}},...,x_{i_{j+1}})$$
are replaced by $T_j(x_1,...,x_k)$.
The formula $\varphi_{j\text{-}\mathit{identities}}$ is
\begin{align*}
\forall x_1\dots\forall x_k\Big(T_j(x_1,\dots, x_{k})
\rightarrow \big(\bigwedge_{(l,m)\in \mathrm{Id}_j} (x_l=x_m)
\wedge \bigwedge_{l,m>i_j} x_l=x_m \big)\Big).
\end{align*}
The formula $\varphi_{j\text{-}\mathit{padding}}$ is the formula
\begin{align*}
\forall x_1\dots\forall x_k\Big(&
\big(R(x_1,\dots,x_k)\rightarrow \forall x_{m_j} T_j(\overline{y}_j, x_{m_j},\dots,x_{m_j})\big)\\
&\wedge \big(\exists x_{m_j} T_j(\overline{y}_j, x_{m_j},\dots,x_{m_j}) \rightarrow  \exists \overline{z}_j R(x_1,\dots,x_k)\big)\Big),
\end{align*}

where $\overline{z}_j$ is the tuple of variables in $(x_1,\dots,x_k)$ but not in $\overline{y}_j$, and ${m_j}\leq k$ is the smallest
integer such that the variable $x_{m_j}$ does not occur in the tuple $\overline{y}_j$; in the case that such variable does not exist the formulas
$\forall x_{m_j} T_j(\overline{y}_j, x_{m_j},\dots,x_{m_j})$ and $\exists x_{m_j} T_j(\overline{y}_j, x_{m_j},\dots,x_{m_j})$ are replaced by $T_j(\overline{y}_j)$.

\item Assume that $tr_k(\psi)=\exists S_1\dots \exists S_n \psi'$ and $tr_k(\vartheta)=\exists T_1\dots \exists T_m\vartheta'$, where $\psi'$ and $\vartheta'$ are $\foc^k$-formulas. Furthermore, assume that the relation variables $S_1,\dots S_n$, $T_1,\dots,T_m$ are all distinct.
\begin{enumerate}
\item If $\varphi$ is of the form $(\psi\vee \vartheta)$, then $tr_k(\varphi)$ is defined to be the formula
\begin{align*}
\exists S\exists T&\exists S_1\dots\exists S_n \exists T_1\dots \exists T_m\Big(\forall x_1\dots \forall x_k \Big(R(x_1,\dots,x_k)\\
&\leftrightarrow \big(S(x_1,\dots,x_k)\vee T(x_1,\dots,x_k)\big)\Big)
\wedge\psi'(S/R)\wedge \vartheta'(T/R)\Big),
\end{align*}
where $\psi'(S/R)$ denotes the formula obtained from $\psi'$ by replacing occurrences of $R$ by $S$,
and analogously for $\vartheta'(T/R)$.
\item If $\varphi=(\psi\wedge \vartheta)$, then $tr_k(\varphi)$ is the formula
$
\exists S_1\dots\exists S_n \exists T_1\dots \exists T_m\big(\psi'\wedge \vartheta'\big).
$
\end{enumerate}
\item If $\varphi$ is of the form $\exists x_i\psi$ and $tr_k(\psi)=\exists S_1\dots\exists S_n \psi'$,
where $\psi'$ is an $\foc^k$-formula, then $tr_k(\varphi)$ is the formula
%
%
%
%
$$\exists S\exists S_1\dots\exists S_n\Big(\forall x_1\dots\forall x_k\big(\exists x_i R(x_1,\dots,x_k)
\leftrightarrow \exists x_i S(x_1,\dots,x_k)\big)
\wedge \psi'(S/R)\Big).$$
%
%
%
%
%
\item If $\varphi$ is of the form $\forall x_i\psi$ and $tr_k(\psi)=\exists S_1\dots\exists S_n \psi'$, where $\psi'$ is an $\foc^k$-formula, then $tr_k(\varphi)$ is defined to be the formula
\begin{align*}
\exists S \exists S_1\dots \exists S_n\Big(&\forall x_1\dots\forall x_k\Big(\big(R(x_1,\dots,x_k)\rightarrow \forall x_i S(x_1,\dots,x_k)\big)\\
&\wedge\big(S(x_1,\dots,x_k)\rightarrow \exists x_i R(x_1,\dots,x_k)\big)\Big)
\wedge \psi'(S/R)\Big).
\end{align*}
\end{enumerate}
A straightforward induction on $\varphi$ shows that for every model $\mA$ and every team
with codomain $A$ such that
$\dom(X)=\{x_1,\dots,x_k\}$,
$\mA\models_X \varphi$ iff $\big(\mA,\rel[X]\big)\models \,tr_k(\varphi)$.
\end{proof}

\begin{theorem}\label{AQtoESOs}
For every $k\geq 1$ and for every $\Sigma_1^1(\foc^k)$-definable atom $A_Q$ it holds that $\fo^k(A_Q)\leq \Sigma^1_1(\foc^k)$, i.e., for every sentence of\, $\fo^k(A_Q)$, there exists an
equivalent sentence of $\Sigma^1_1(\foc^k)$.
\end{theorem}
\begin{proof}
Let $\tau$ be a relational vocabulary, $k\geq 1$, and $A_Q$ a $\Sigma_1^1(\foc^k)$-definable atom. Let $\varphi$ be an $\fo^k(A_Q)(\tau)$-sentence and
$\varphi^*=\exists R_1\dots \exists R_n \psi$
the related $\Sigma_1^1(\foc^k)(\tau\cup\{R\})$-sentence given by Lemma \ref{DtoESO}. The following conditions are equivalent.
\begin{enumerate}
\item $\mA\models \varphi$.\label{deso1}
\item $\mA\models_X \varphi$ for some nonempty team $X$ such that $\dom(X)=\{x_1,\dots,x_k\}$.\label{deso2}
\item $\big(\mA,\rel[X]\big)\models \varphi^*$ for some nonempty team $X$ such that $\dom(X)=\{x_1,\dots,x_k\}$.\label{deso3}
\item $(\mA, R)\models \exists R_1\dots \exists R_n \big(\exists x_1\dots\exists x_k R(x_1,\dots,x_k)\wedge\psi\big)$ for some $R\subseteq A^k$.\label{deso4}
\item $\mA\models \exists R\exists R_1\dots \exists R_n \big(\exists x_1\dots\exists x_k R(x_1,\dots,x_k)\wedge\psi\big)$.\label{deso5}
\end{enumerate}
The equivalence of \ref{deso1} and \ref{deso2} follows from Proposition \ref{freevar} and the fact that $\fr(\varphi)=\emptyset$. By Lemma \ref{DtoESO}, conditions \ref{deso2} and \ref{deso3} are equivalent. The equivalence of \ref{deso3} and \ref{deso4} follows from the fact that $\varphi^*=\exists R_1\dots \exists R_n \psi$.
The conditions \ref{deso4} are \ref{deso5} clearly equivalent.
\end{proof}

\begin{theorem}\label{AQtwo nexptime}
Let $A_Q$ be a $\Sigma_1^1(\foc^2)$-definable generalized atom. Then the problems
$\sat[\fotwo(A_Q)]$ and $\finsat[\fotwo(A_Q)]$ are \NEXPTIME-complete.
\end{theorem}
\begin{proof}
Since the translation $\varphi\mapsto \varphi^*$ is
computable in polynomial time and (finite) satisfiability of
$\Sigma_1^1(\foc^2)$ can be checked in \NEXPTIME \cite{IEEEonedimensional:pratth},
we conclude that both $\sat[\fotwo(A_Q)]$ and $\finsat[\fotwo(A_Q)]$ are in $\NEXPTIME$.
On the other hand, since $\fotwo\le \fotwo(A_Q) $ by Proposition \ref{whatever},
and since both $\sat[\fotwo]$
and $\finsat[\fotwo]$ are \NEXPTIME-hard
\cite{grkova97}, it follows that both \sat[\fotwo(A_Q)] and also \finsat[\fotwo(A_Q)] are as well.
\end{proof}
The result of Theorem \ref{AQtwo nexptime} can be directly generalized
to concern finite collections $\mathcal{A}$ of generalized atoms.
The proof of the following theorem is practically the same as that of Theorem \ref{AQtwo nexptime}.
\begin{theorem}\label{classNEXPTIME}
Let $\mathcal{A}$ be a finite collection of\, $\Sigma_1^1(\foctwo)$-definable generalized atoms.
%
%
%
The satisfiability and the finite satisfiability problems of\, $\fotwo(\mathcal{A})$ are $\NEXPTIME$-complete.
\end{theorem}

We shall next make use of Theorem \ref{classNEXPTIME} in order to show that the satisfiability
and the finite satisfiability problems of two-variable fragments of dependence logic,
inclusion logic, exclusion logic and independence
 logic are $\NEXPTIME$-complete. The result for two-variable dependence logic was already established in
\cite{kokulovi2}. Note that when regarded as generalized atoms,
each of the dependency notions above correspond to a collection of generalized atoms; for example the 
atomic formulas $\dep[x,y]$ and $\dep[x,y,z]$ refer to two different atoms,
one of type $(2)$ and the other of type $(3)$.
However, in order to capture the two-variable fragments of of these logics, we only need a finite number of
generalized atoms for each logic, as we shall see.
We define $\varphi_{const}:= \exists^{\leq 1} x R(x)$, $\varphi_{dep} := \forall x \exists^{\leq 1} y R(x,y)$,
$
\varphi_{inc} := \forall x \forall y \big(R(x,y)\rightarrow S(x,y)\big), \quad \varphi_{exc} := \forall x \forall y \big(R(x,y)\rightarrow \neg S(x,y)\big),
\varphi_{ind} := \forall x \forall y \big((\exists y R(x,y) \wedge \exists x R(x,y)) \rightarrow R(x,y)\big).
$

The formulas $\varphi_{const}$, $\varphi_{dep}$, $\varphi_{inc}$, $\varphi_{exc}$ and $\varphi_{ind}$  define the generalized atoms $A_{const}$ of type $(1)$, $A_{dep}$ of type $(2)$, $A_{inc}$ of type $(2,2)$, $A_{exc}$ of type $(2,2)$, and $A_{ind}$ of type $(2)$, respectively. 

%
%
%

\begin{theorem}\label{some}
The satisfiability and finite satisfiability problems of 
the two-variable fragments of dependence logic,
inclusion logic, exclusion logic, inclusion/exclusion logic, and independence logic are all $\NEXPTIME$-complete.
\end{theorem}
\begin{proof}
We establish polynomial time translations $\dtwo \to \fotwo(\{A_{const}, A_{dep}\})$, $\inc^2 \to \fotwo(A_{inc})$, $\exc^2 \to \fotwo(A_{exc})$, $\incexc^2 \to \fotwo(A_{inc},A_{exc})$, and  $\ind^2 \to \fotwo(\{A_{const}, A_{dep}, A_{ind}\})$ that preserve equivalence. The result then follows from Theorem \ref{classNEXPTIME} and the fact that the generalised atoms $A_{const}$, $A_{dep}$, $A_{exc}$, $A_{inc}$, $A_{ind}$ are all $\Sigma_1^1(\foctwo)$-definable.

Notice first that in dependence atoms, repetition of variables can always be avoided.
The atom $\dep[\overline{x},y]$ is equivalent to the atom $\dep[\overline{x}\, ',y]$, where $\overline{x}\, '$ is obtained from
$\overline{x}$ by
simply removing the repetition of variables. Furthermore, if $y$ occurs in the tuple $\overline{x}$,
then $\dep[\overline{x},y]$ is equivalent to $y=y$. Thus we may assume that in formulas of two-variable dependence logic, only dependence atoms $\dep[x]$, $\dep[y]$, $\dep[x,y]$, and $\dep[y,x]$ may occur. Clearly $\dep[x]$
is equivalent to the generalized atom $A_{const}(x)$,
while $\dep[x,y]$
is equivalent to the generalized atom $A_{dep}(x,y)$.
Since $A_{const}$ and $A_{dep}$ are $\Sigma_1^1(\foctwo)$-definable atoms, by Theorem \ref{classNEXPTIME},  $\sat[\fotwo(\{A_{const}, A_{dep}\})]$ and $\finsat[\fotwo(\{A_{const}, A_{dep}\})]$ are $\NEXPTIME$-complete. Thus both $\sat[\dtwo]$ and $\finsat[\dtwo]$ are as well.

%
It is straightforward to show that in two-variable inclusion logic, only inclusion atoms of type
$(y_1,y_2)\subseteq(z_1,z_2)$, where $y_1,y_2,z_1,z_2\in \{x,y\}$, are needed.
For example, the inclusion atom $x\subseteq y$ can be replaced by the equivalent inclusion atom $(x,x)\subseteq (y,y)$,
and the inclusion atoms $(x,y,x)\subseteq (x,y,y)$ and $(x,y,y)\subseteq (y,x,x)$
can be replaced by the equivalent atomic formulas $x=y$
and $(x,y)\subseteq (y,x)$, respectively.
Thus we may assume that in formulas of two-variable inclusion logic,
only inclusion atoms of type $(y_1,y_2)\subseteq(z_1,z_2)$ may occur;
inclusion atoms of other kinds can easily be eliminated in polynomial time.
 Clearly $(y_1,y_2)\subseteq(z_1,z_2)$ is equivalent to the generalized atom $A_{inc}\big((y_1,y_2),(z_1,z_2)\big)$. Since $A_{inc}$ is a $\Sigma_1^1(\foctwo)$-definable atom, it follows from Theorem \ref{AQtwo nexptime} that $\sat[\fotwo(A_{inc})]$ and $\finsat[\fotwo(A_{inc})]$ are $\NEXPTIME$-complete. Thus $\sat[Inc^2]$ and $\finsat[Inc^2]$ are as well.

Using analogous argumentation, it is straightforward to show that in two-variable exclusion logic, only exclusion atoms of type $(y_1,y_2)\mid(z_1,z_2)$, where $y_1,y_2,z_1,z_2\in \{x,y\}$, are needed. Clearly $(y_1,y_2)\mid(z_1,z_2)$ is equivalent to the generalized atom $A_{exc}\big((y_1,y_2),(z_1,z_2)\big)$. Since $A_{exc}$ is a $\Sigma_1^1(\foctwo)$-definable atom, it follows from Theorem \ref{AQtwo nexptime} that both $\sat[\fotwo(A_{exc})]$ and $\finsat[\fotwo(A_{exc})]$ are $\NEXPTIME$-complete. Thus $\sat[\exc^2]$ and $\finsat[\exc^2]$ are as well. Similarly it follows that $\sat[\incexc^2]$ and $\finsat[\incexc^2]$ are $\NEXPTIME$-complete.
Likewise, it is easy to show that
in the formulas of two-variable independence logic,
only restricted versions of independence atoms are needed.
First notice that we may always assume that in independence atoms
$\overline{x}\bot_{\overline{y}}\overline{z}$, repetition of variables does not occur in any of the tuples $\overline{x}$,
$\overline{y}$ and $\overline{z}$.
By the semantics of independence atoms, it is also easy to check that the atoms $\overline{x} \bot_{\overline{y}} \overline{z}$ and $\overline{z}
\bot_{\overline{y}} \overline{x}$ are always equivalent. Furthermore, it is clear that the order of variables in the tuples $\overline{x}$, $\overline{y}$, and $\overline{z}$ makes no difference. Notice then that each of the
following atoms in the variables $x,y$ is equivalent to the formula $\exists x \, x=x$:
\[
\emptyset\bot_{\overline{x}} \overline{y}, \quad \overline{x}\bot_{(x,y)}\overline{y}, \quad x\bot_x x, \quad x\bot_x y, \quad x\bot_x (x,y), \quad y\bot_y y, \quad x\bot_y y, \quad y\bot_y (x,y).
\]
Notice also the following equivalences:
\begin{align*}
&(x,y)\bot_x (x,y) \equiv y\bot_x y, \quad y\bot_x (x,y)\equiv y\bot_x y,\quad (x,y)\bot_y (x,y) \equiv x\bot_y x,\\
&x\bot_y (x,y)\equiv x\bot_y x, \quad x\bot (x,y)\equiv x\bot x, \quad y\bot (x,y)\equiv y\bot y.
\end{align*}
Thus we may assume that only the independence atoms
$x\bot x$, $y\bot y$, $x\bot y$, $(x,y)\bot (x,y)$, $x\bot_y x$, and $y\bot_x y$ occur in the formulas of two-variable independence logic. It is straightforward to check that the following equivalences between independence atoms and generalized atoms hold:
\begin{align*}
& x\bot x \equiv A_{const}(x), \quad y\bot y \equiv A_{const}(y), \quad x\bot y \equiv A_{ind}\big((x,y)\big),\\
& (x,y)\bot (x,y) \equiv A_{const}(x)\wedge A_{const}(y), \quad x\bot_y x \equiv A_{dep}(y,x), \quad y\bot_x y \equiv A_{dep}(x,y).
\end{align*}
Since $A_{const}$, $A_{dep}$, and $A_{ind}$ are all $\Sigma_1^1(\foctwo)$-definable atoms, it follows from Theorem \ref{classNEXPTIME} that $\sat[\fotwo(\{A_{const}, A_{dep}, A_{ind}\})]$ and $\finsat[\fotwo(\{A_{const}, A_{dep}, A_{ind}\})]$ are $\NEXPTIME$-complete. Thus $\sat[\ind^2]$ and $\finsat[\ind^2]$ are as well.
\end{proof}

\section{Undecidability via non-tiling}\label{tiling}
%
In this section we introduce structures and methods that we will later employ to prove undecidability of the validity problem of two-variable dependence logic. 
Curiously, all attempts (by us or known to us) to use the
standard ($\Pi_1^0$-complete) tiling problem for the
undecidability proof have failed;
we will instead use the ($\Sigma_1^0$-complete) non-tiling
problem in our arguments below.
\emph{The grid} is the structure $\mG=(\mathbb{N}^2,V,H)$, where
$V=\{\big((i,j), (i,j+1)\big)\in \mathbb{N}^2\times \mathbb{N}^2 \mid i,j\in \mathbb{N}\}$  and 
$H=\{\big((i,j), (i+1,j)\big)\in \mathbb{N}^2\times \mathbb{N}^2 \mid i,j\in \mathbb{N}\}$.
%
%
%
%
%
%
%
A function $t:4\longrightarrow\mathbb{N}$ is called a \emph{tile type}.
Define the set
\(
\mathrm{TILES} := \{ P_t \mid t\text{ is a tile type} \}
\)
of unary relation symbols.
The unary relation symbols in the set $\mathrm{TILES}$ are called \emph{tiles}.
%
%
The number $t(0)$ is the \emph{top colour}, $t(1)$ the \emph{right colour},
$t(2)$ the \emph{bottom colour}, and $t(3)$ the \emph{left colour} of $P_t$.
Let $T$ be a finite nonempty set of tiles and $V$ and $H$ binary relation symbols. We say that a structure $\mf{A}=(A,V,H)$ is
\emph{$T$-tilable}, if there exists an expansion of $\mf{A}$ to
the vocabulary
$\{H,V\}\cup\{\ P_t\ |\ P_t\in T\ \}$
such that the following conditions hold for all $u,v \in A$.
\begin{enumerate}
\item
The point $u$ belongs to the extension of
exactly one symbol $P_t$ in $T$.
\item
If $u H v$, $P_t(u)$ and $P_s(v)$,
then the right colour of $P_t$ is the same as the left
colour of $P_{s}$.
\item
If $u V v$, $P_t(u)$ and $P_s(v)$,
then the top colour of $P_t$ is the same as the bottom
colour of $P_{s}$.
\end{enumerate}
We will next define the \emph{tiling problem} and the \emph{non-tiling problem}.
Let $\mathcal{F}$ denote the set of finite, nonempty subsets of $\mathrm{TILES}$.
We define 
$\mathcal{T} := \{T\in\mathcal{F} \mid \mf{G}\text{ is $T$-tilable} \}$ and ${\mathcal{\bar{T}'}} := \{T\in\mathcal{F} \mid \mf{G}\text{ is not $T$-tilable} \}$.
The \emph{tiling problem (non-tiling problem,} resp.) is the membership problem of the set $\mathcal{T}$ ($\mathcal{\bar{T}'}$, resp.)
with the input set $\mathcal{F}$.
\begin{theorem}[\cite{berger}]
The tiling problem is $\Pi_1^0$-complete.
\end{theorem}
The \emph{non-tiling problem} is the complement of the tiling problem. Thus the following corollary follows.
\begin{corollary}\label{nontilingc}
The non-tiling problem is $\Sigma_1^0$-complete.
\end{corollary}
The proof of the following lemma is straightforward.
\begin{lemma}\label{tilingdefinablelemma}
There is a computable function associating
each input $T$ to the
non-tiling problem with an $\fotwo$-sentence $\varphi_{T}$ of
the vocabulary $\tau:=\{H,V\}\cup T$ 
such that for every structure $\mf{A}$ of the vocabulary $\{H,V\}$,  the structure $\mf{A}$ is not $T$-tilable iff
for every expansion $\mf{A}^*$ of\, $\mf{A}$ to the
vocabulary $\tau$, it holds that $\mf{A}^*\models\varphi_{T}$.
\end{lemma} 
%
%
%
%
%
%
%

\begin{definition}\label{def:gridlike}
Let $\tau=\{V,H\}$ be a vocabulary where $V$ and $H$ are binary relation symbols. Let  $\mA=(A,V,H)$ be a $\tau$-structure. We say that $\mA$ is \emph{gridlike} if the below conditions hold.
\begin{enumerate}
\item The extension of $V$ in $\mA$ is serial (i.e., $\forall x\in A$ $\exists y\in A$ s.t. $V(x,y)$).
\item The extension of $H$ in $\mA$ is serial (i.e., $\forall x\in A$ $\exists y\in A$ s.t. $H(x,y)$).
\item If $a,b,c,b',c'\in A$ are such that $V(a,b)$, $H(b,c)$, $H(a,b')$, and $V(b',c')$, then $c=c$'.
\end{enumerate}
\end{definition}

Note that it follows from the above definition that in gridlike structures, for every
point $a$, there exist points $b$, $c$ and $d$ such that $H(a,b)$, $V(a,c)$, $V(b,d)$, and $H(c,d)$.

Let $\tau$ be the vocabulary of gridlike structures and $U$, $P$, $Q$, $C$ unary relation symbols. We say that a $\tau\cup\{U, P, Q, C\}$-structure $\mA$ is \emph{striped and gridlike} if the $\tau$-reduct of $\mA$ is gridlike, the extensions of $P$ and $Q$ in $\mA$ are \emph{distinct} singleton sets, the extension of $U$ in $\mA$ is the union of the extensions of $P$ and $Q$, and $\mA$ has the following property (intuitively $C$ creates stripes in $\mA$):
\begin{equation}\label{ideaofC}
\text{$\bigl(H(a,b)\Rightarrow (C(a)\Leftrightarrow C(b))\bigl)$
and $\bigl(V(a,b)\Rightarrow (C(a)\Leftrightarrow \neg C(b))\bigl)$.}
\end{equation}
The following lemma can be now proven by a simple inductive argument.
\begin{lemma}\label{homolemma}
If $\mA$ is striped and gridlike, then there exists a homomorphism from the grid into $\mA$. 
\end{lemma}
%

%
\begin{lemma}\label{likestogrids}
Let $T$ be an input to the non-tiling problem. The grid is non-$T$-tilable iff  (the \{H,V\}-reduct of) every striped gridlike structure is non-$T$-tilable.
\end{lemma}
\begin{proof}
The direction from left to right follows from Lemma \ref{homolemma} in a straightforward way.
The converse holds since the grid is an \{H,V\}-reduct of a striped gridlike structure.
\end{proof}

\section{The validity problem of $\dtwo$ is undecidable}

In this section we give a reduction from the non-tiling problem to the validity problem of $\dtwo$.

Let $\tau=\{V,H,C,U,P,Q\}$ be the vocabulary of striped gridlike structures.
We will first define a formula $\varphi_{\mathit{non-grid}}$ of $\dtwo$ such that $\mA$ is not striped and gridlike iff $\mA\models \varphi_\mathit{non-grid}$.
%
%
%
We first notice that the first two conditions of Definition \ref{def:gridlike} are easy to deal with. Define
\(
\varphi_{\mathit{non-serial}}  := \exists x \forall y \neg V(x,y) \lor \exists x \forall y \neg H(x,y).
\)
The third condition of Definition \ref{def:gridlike} is nontrivial.
In the below construction, we will use the predicates $P$, $Q$, $U$ for counting (only). We will first show how to force the extensions of $P$ and $Q$ to be distinct singletons and the extension of $U$ to be the union of $P$ and $Q$. The next formulae will be used for dealing with the cases where this \emph{does not} hold.
\begin{align*}
\varphi_{\mathit{non-singleton}}(X)  :=\ & \forall x \neg X(x) \lor \exists x \exists y \big(X(x)\wedge X(y)\wedge \neg x=y\big)\\
\varphi_{\mathit{non-distinct}}(X,Y)  :=\ & \exists x \big(X(x) \land Y(x)\big)\\
\varphi_{\mathit{non-union}}(X,Y,Z)  :=\ & \exists x \Big(X(x) \land \big(\neg  Y(x)\lor \neg Z(x)\big)\Big) \lor \exists x \Big(\neg X(x) \land \big(Y(x)\lor Z(x)\big)\Big)\\
\varphi_{\mathit{\lvert U \rvert \not= 2}}  :=\ & \varphi_{\mathit{non-singleton}}(P) \lor \varphi_{\mathit{non-singleton}}(Q) \lor \varphi_{\mathit{non-distinct}}(P,Q)\\
&\lor \varphi_{\mathit{non-union}}(U,P,Q).
\end{align*}
It is easy to check that the $\tau$-models $\mA$ such that $\mA\not\models \varphi_{\mathit{\lvert U \rvert \not= 2}}$ are exactly those models where the extensions of $P$ and $Q$ are distinct singletons and the extension of $U$ is the union of the extensions of $P$ and $Q$ (and thus the cardinality of 
the extension of $U$ is $2$).
%

%
%
%
%
%
%
%
%
%
We will now show how to enforce Equation \eqref{ideaofC}.
The formula $\varphi_{\mathit{non-stripes}}$ below takes care of the
cases where  \eqref{ideaofC} does \emph{not} hold. Define
\[
\varphi_{\mathit{non-stripes}}  := \exists x \exists y \Big(\Big(H(x,y) \wedge \big(C(x) \leftrightarrow \neg C(y)\big) \Big) \lor \Big( V(x,y) \wedge \big( C(x) \leftrightarrow C(y)\big)\Big)\Big).
\]
We are now ready to show how to deal with models that
violate the last condition of Definition \ref{def:gridlike}. To understand the intended
meaning of  the following formula, assume that the extension of $U$ is of size two and that the condition given by Equation \eqref{ideaofC} holds. Note also that  from \eqref{ideaofC} it follows that if such points $c$ and $c'$ exist that violate the last
condition of Definition \ref{def:gridlike}, then $c$ and $c'$ agree about $C$, i.e., 
we have $C(c)$ iff $C(c')$. We first deal with the case where $C(c)$ and $C(c')$ both hold. We denote by $\varphi_{\mathit{non-C^+-join}}$ the following formula (whose meaning is fully explained in 
the proof of Lemma \ref{joinlemma}):
\begin{align*}
 & \forall x \Big(\neg U(x) \lor \exists y  \Big( C(y) \land \dep[y,x]
\land \exists x \Big(\dep[x,y] \land \big(\big(\dep[x]\wedge H(x,y) \big) \lor \big(\dep[x] \wedge V(x,y)\big)\big) \\
&\quad\quad\quad \land \exists y \big(\dep[y] \land \big(V(y,x) \lor H(y,x)\big) \land \neg C(y))\big)\Big)\Big)\Big).
%
%
%
\end{align*}
To deal with the case where $\neg C(c)$ and $\neg C(c')$,
we define the formula $\varphi_{\mathit{non-C^--join}}$ which is obtained from $\varphi_{\mathit{non-C^+-join}}$ by simultaneously replacing each $C(x)$ and $C(y)$ by $\neg C(x)$ and $\neg C(y)$, respectively. Finally, we define that $\varphi_{\mathit{non-join}}\dfn \varphi_{\mathit{non-C^+-join}} \lor \varphi_{\mathit{non-C^--join}}$ and $\varphi_{\mathit{non-grid}}  :=  \varphi_{\mathit{non-serial}} \lor  \varphi_{\mathit{\lvert U \rvert \not= 2}} \lor \varphi_{\mathit{non-stripes}} \lor \varphi_{\mathit{non-join}}$.
%
%
%
%
%
%
%
%
%
\begin{lemma}\label{joinlemma}
Let $\tau=\{V,H,C,U,P,Q\}$ be the vocabulary of striped gridlike structures. Let $\mA$ be a $\tau$-structure such that the extension of $U$ is of cardinality $2$. Assume the condition \eqref{ideaofC} holds. Then $\mA\models \varphi_{\mathit{non-join}}$
iff the last condition of Definition \ref{def:gridlike} fails in $\mA$.
\end{lemma}
\begin{proof}
From \eqref{ideaofC} it follows that if such $c$ and $c'$ exist in $\mA$ that violate the last condition of Definition \ref{def:gridlike}, then $c$ and $c'$ agree on $C$.
We will show that
\begin{align}\label{meetcase}
\empty\hspace{-5mm}\vspace{-3mm}
\mA\models\varphi_{\mathit{non-C^+-join}}  \text{ iff the last condition of Def. \ref{def:gridlike} fails in $\mA$ 
for some $c,c'$ s.t. $C(c)$ \& $C(c')$}.
\end{align}
The analogous argument for $\varphi_{\mathit{non-C^--join}}$
and the case where $\neg C(c)$ and $\neg C(c')$ hold is similar.
Below we denote by $\{(x_1,v_1),...,(x_k,v_k)\}$ the 
variable assignment that maps $x_i$ to $v_i$ for each $i$.
Let $u,u'$ be the elements that are in the extension of $U$ in $\mA$.
We thus have $\mA\models\varphi_{\mathit{non-C^+-join}}$ iff
\begin{align*}
\hspace{-4mm}\mA \models_{X_1}  \exists y&  \Big( C(y) \land \dep[y,x] 
\land \exists x \Big(\dep[x,y] \land \big(\big(\dep[x]\wedge H(x,y) \big)\\
&\quad \lor \big(\dep[x] \wedge V(x,y)\big)\big)  \land \exists y \big(\dep[y] \land \big(V(y,x) \lor H(y,x)\big) \land \neg C(y))\big)\Big)\Big),
\end{align*}
where $X_1= \{\{(x,u)\},\{(x,u')\}\}$.
Now, recalling that dependence logic has the downwards closure property (cf. 
proposition \ref{Downward closure}), we observe that 
the above holds if and only if there exist \emph{distinct} (distinctness being due to 
the atom $\dep[y,x]$)  points $c, c'$ in the extension of $C$ such that
\begin{align*}
\hspace{-3mm}\mA \models_{X_2} 
\,\exists x \Big(\dep[x,y] \land \big(\big(\dep[x]&\wedge H(x,y) \big) \lor \big(\dep[x] \wedge V(x,y)\big)\big) \\
&\quad\quad\quad \land \exists y \big(\dep[y] \land \big(V(y,x) \lor H(y,x)\big) \land \neg C(y))\big)\Big),
\end{align*}
where $X_2= \{\{(x,u),(y,c)\},\{(x,u'),(y,c')\}\}$.
The above holds if and only if there exist distinct points $b, b'$ of $\mA$ such that
$H(b,c)$ and $V(b',c')$ (or $V(b,c)$ and $H(b',c')$ in which case the argument is analogous) and
\begin{align*}
\mA \models_{X_3} &
\,\exists y \big(\dep[y] \land \big(V(y,x) \lor H(y,x)\big) \land \neg C(y))\big),
\end{align*}
where $X_3=\{\{(x,b),(y,c)\}, \{(x,b'),(y,c')\}\}$.
The above holds if and only if there exists a point $a$ in $\mA$ such that $\neg C(a)$, ($V(a,b)$ or $H(a,b)$) and ($V(a,b')$ or $H(a,b')$). Since $C(c)$ and $C(c')$ hold, it follows from the assumption that \eqref{ideaofC} holds that $C(b)$ and $\neg C(b')$. Now since also $\neg C(a)$ holds, it follows again from
\eqref{ideaofC} that $V(a,b)$ and $H(a,b')$.
When all of the above is combined, we obtain \eqref{meetcase}.
The analogous condition where $\neg C(c)$ and $\neg C(c')$ is
proved similarly. Since \eqref{ideaofC} holds for $\mathfrak{A}$, any
points $c$ and $c'$ of $\mA$ that violate the last condition of Definition
\ref{def:gridlike}, must agree on $C$. Thus the lemma holds.
\end{proof} 

The next lemma follows from Lemma \ref{joinlemma} together with the observations made earlier in this section.
\begin{lemma}\label{nongridlemma}
Let $\tau=\{V,H,C,U,P,Q\}$ be the vocabulary of striped gridlike structures and
let $\mA$ be a $\tau$-model.
%
%
%
Then $\mA$ is striped
and gridlike iff $\mA\not\models \varphi_\mathit{non-grid}$.
%
%
\end{lemma}

\begin{theorem}
The validity problem for $\dtwo$ is undecidable (more precisely, $\Sigma^0_1$-hard).
\end{theorem}
\begin{proof}
We give a computable reduction from the non-tiling problem to the validity problem of $\dtwo$. Since the former is $\Sigma^0_1$-complete (Corollary \ref{nontilingc}), we obtain $\Sigma^0_1$-hardness for the latter.

If $T$ is an input to the non-tiling problem, then $\varphi_T$ denotes the $\fotwo$-sentence given by Lemma \ref{tilingdefinablelemma} and $\varphi_{\mi{non-T-tiling}}\dfn(\varphi_{\mathit{non-grid}}\lor \varphi_T)$. 
Let $\tau$ be as defined in Lemma \ref{nongridlemma}.
Let $C_{\tau,T}$ denote the class of all $\tau\cup T$-structures and let $\cC^{\tau,T}_\mi{s-gridlike}$ be the class of exactly all expansions of striped gridlike structures to the vocabulary $\tau\cup T$.  
Let $T$ be an input to the non-tiling problem. We will show that the grid is non-T-tilable iff the $\dtwo$-sentence  $\varphi_{\mi{non-T-tiling}}$ is valid.
By definition,
$\varphi_{\mi{non-T-tiling}}$ is valid iff
\(
\mA \models \varphi_{\mathit{non-grid}}\lor \varphi_T \text{ holds for every $\mA\in \cC_{\tau,T}$.}
\)
Since $\varphi_{\mathit{non-grid}}$ and $\varphi_T$ are sentences, the right-hand side of this equivalence is equivalent to the claim that 
\begin{equation}\label{sglike0}
\forall \mA\in \cC_{\tau,T}: \mA \models \varphi_{\mathit{non-grid}} \text{ or  } \mA \models \varphi_{T}.
\end{equation}
By Lemma \ref{nongridlemma}, $\mB^*\models \varphi_{\mathit{non-grid}}$ holds for every  $\tau$-reduct $\mB^*$ of $\mB\in \cC_{\tau,T}$ that is not striped and gridlike. Hence for every  $\mB\in \cC_{\tau,T}$ such that the  $\tau$-reduct $\mB^*$ of $\mB$ is not striped and gridlike, it holds that  $\mB\models \varphi_{\mathit{non-grid}}$. Thus \eqref{sglike0} is equivalent to the claim that
\begin{equation}\label{sglike1}
\forall \mA\in \cC^{\tau,T}_\mi{s-gridlike}: \mA \models \varphi_{\mathit{non-grid}} \text{ or  } \mA \models \varphi_{T}.
\end{equation}
Now let $\mB$ be an arbitrary striped and gridlike $\tau$-structure. By Lemma \ref{nongridlemma}, $\mB\not\models \varphi_{\mathit{non-grid}}$. Thus $\mB^*\not\models \varphi_{\mathit{non-grid}}$ for every expansion $\mB^*$ of $\mB$ to the vocabulary $\tau\cup T$. From this it follows that \eqref{sglike1} is equivalent to the claim that
\begin{equation}\label{sglike2}
\forall \mA\in \cC^{\tau,T}_\mi{s-gridlike}: \mA \models \varphi_{T}.
\end{equation}
Thus, by Lemma \ref{tilingdefinablelemma}, \eqref{sglike2} holds if and only if every striped gridlike structure is non-$T$-tilable. Finally, from Lemma \ref{likestogrids} it follows that this is equivalent to the claim that the grid is non-$T$-tilable.
\end{proof}

\section{Satisfiability of $\exists^*\forall^*$-formulas}\label{S4}

\vspace{-2mm}

In this section we consider the complexity of  satisfiability for  sentences  of dependence logic and its variants in the prefix class $\exists^*\forall^*$. 
For first-order logic, the satisfiability and finite satisfiability problems of  the  prefix class $\exists^*\forall^*$
are known to be  $\NEXPTIME$-complete.
The results hold for both the case with equality and the case without equality, see \cite{graadeli}.
%

Let $\mathcal{A}$ be a collection of generalized atoms. We denote by $\exists^*\forall^* [\mathcal{A}]$ the class of sentences of $\fo(\mathcal{A})$ of the form
$
\exists x_0\cdots \exists x_n \forall y_0\cdots \forall y_m \theta,
$
where $\theta$ is a quantifier-free formula whose generalized atoms are in $\mathcal{A}$. 
%
It is worth noting that, depending on the set  $\mathcal{A}$, the expressive power and complexity of sentences in  $\exists^*\forall^* [\mathcal{A}]$ can vary considerably even when $\mathcal{A}$ is finite and contains
only computationally non-complex atoms.
For example, there are universal sentences of dependence logic that define NP-complete problems \cite{Kontinen13}. Furthermore,  every sentence of inclusion logic is equivalent to a sentence with a prefix of the form  $\exists^*\forall^1$ \cite{Hannula14} implying that
the satisfiability problem of the $\exists^*\forall^*$-fragment of inclusion logic is undecidable.


Recall that we say that  a formula $\varphi$ is \emph{closed under substructures} if for all $\mathfrak{A}$ and $X$
it holds that if
$\mathfrak{A}\models_X \varphi $, $\mathfrak{A}':=\mathfrak{A}\upharpoonright B$ and $X':=X\upharpoonright B$ for some
$B\subseteq A$, then we have $\mathfrak{A}'\models_{X'} \varphi $. 
%
\begin{lemma}\label{AE-lemma} Let $\mathcal{A}$ be a collection of generalized  atoms that are closed under substructures.  Then the following conditions hold.
\begin{enumerate}
\item\label{A} Suppose  $\varphi\in \fo[\mathcal{A}]$ is of the form $\forall y_0\cdots \forall y_m \theta$, where $\theta$ is quantifier-free. Then $\varphi$  is closed under substructures.     
\item\label{EA} Let $\varphi\in \exists^*\forall^*[\mathcal{A}]$ be a sentence. Then, if $\varphi$ is satisfiable, $\varphi$ has a model with at most $max\{1,k\}$ elements, where $k$ refers to the number of existentially quantified variables in $\varphi$.  
\end{enumerate}
\end{lemma}

\begin{proof}
We will first prove claim \eqref{A}. Suppose that $\varphi:=\forall y_0\cdots \forall y_m \theta$. We will first show the claim for quantifier-free formulas $\theta$, i.e., we will show that for all  $\mathfrak{A}$, $X$, $\mathfrak{A}'$, and  $X'$
such that $\mathfrak{A}':=\mathfrak{A}\upharpoonright B$ and $X':=X\upharpoonright B$ for some $B\subseteq A$,
%
%
%
the following implication holds.
\begin{equation}\label{theta}
\mA\models_{X} \theta \Rightarrow \mA'\models_{X'}\theta.
\end{equation}
The claim obviously holds if $\theta$ is a first-order literal. If $\theta$ is a generalized atom from $\mathcal{A}$,
then the claim holds by assumption.
The case $\theta := \psi_1\wedge \psi_2$ follows immediately from the induction hypothesis.
Let us then assume that  $\theta := \psi_1\vee \psi_2$. Since $\mA\models_{X} \theta$, there are sets $Y$ and $Z$ such that
 $Y\cup Z=X$, $\mA\models_{Y} \psi_1$ and $\mA\models_{Z} \psi_2$.  By the induction hypothesis, we have 
$\mA'\models_{Y'} \psi_1$ and $\mA'\models_{Z'} \psi_2$, where $Y':=Y\upharpoonright B$  and  $Z':=Z\upharpoonright B$.
Since $Y'\cup Z'=X'$, it follows that $\mA'\models_{X'} \theta$.
We will now show that the claim also holds for $\varphi$. Suppose that $\mA\models_{X} \varphi$.
Then, by the truth definition, 
 \(\mA\models_{X[A/y_0]\cdots [A/y_m]}\theta. \)
\noindent Using \eqref{theta}, we have 
  \(\mA'\models_{(X[A/y_0]\cdots [A/y_m])\upharpoonright B}\theta. \)
It is easy to check that 
\( (X[A/y_0]\cdots [A/y_m])\upharpoonright B= (X\upharpoonright B) [B/y_0]\cdots [B/y_m].\)
Hence we have
\(\mA'\models_{X'}\varphi. \)

Let us then prove \ref{EA}.
Assume $\varphi$ is a sentence of the form
$\exists x_0\cdots \exists x_n \forall y_0\cdots \forall y_m \theta,$
where $\theta$ is quantifier-free,
and that there is a structure  $\mA$ such that $\mA\models \varphi$. Hence there exists functions $F_i$ such that
$\mA\models_ X \forall y_0\cdots \forall y_m \theta,$
where $X=\{\emptyset\}[F_0/x_0]\cdots [F_n/x_n]$. Let $s$ be some assignment in $X$.
Let $\rng(s)$ denote the set of elements $b$ such that 
$s(x) = b$ for some variable $x$ in the domain of $s$.
If $\rng(s)\not=\emptyset$ define $B := \rng(s)$, and if
$\rng(s) = \emptyset$ (i.e., $s = \emptyset$), define $B=\{b\}$, where $b$ is
an arbitrary element in $A$.
By claim \eqref{A}, the formula
$\forall y_0\cdots \forall y_m \theta$ is closed under substructures. Thus
$
\mA\upharpoonright B \models_{X\upharpoonright B} \forall y_0\cdots \forall y_m \theta.
$
Thus it follows that
\(
\mA\upharpoonright B \models \varphi.
\)
\end{proof}

A generalized atom $A_{Q}$ is said to be \emph{polynomial time computable} if the question whether $\mathfrak{A}\models_X A_{Q}(\overline{y}_1,...,\overline{y}_n)$ holds can be decided in time polynomial in the size of $\mA$ and $X$.
A class of atoms $\mathcal{A}$ is said to be \emph{uniformly polynomial time computable} if there exists a polynomial function $f:\mathbb{N}\to\mathbb{N}$ such that for every atom $A_Q\in \mathcal{A}$ it holds that the question whether $\mathfrak{A}\models_X A_{Q}(\overline{y}_1,...,\overline{y}_n)$ holds can be
decided in time $f\bigl(\lvert\mA\rvert+\lvert X\rvert + \vert A_Q(\overline{y}_1,...,\overline{y}_n)\vert\, \bigr)$.
Note that every finite class of polynomial time computable atoms is also uniformly polynomial time computable.
%
%

%
The following theorem now follows from Lemma \ref{AE-lemma}. We will make use of
the recent result of Gr\"adel showing that for a
uniformly polynomial time computable collection $\mathcal{A}$ of atoms, 
the model checking problem for $\fo (\mathcal{A})$-formulas  is in  $\NEXPTIME$ \cite{DBLP:journals/tcs/Gradel13}.
\begin{theorem}\label{EAcomplete} Let  $A_{Q}$  be a generalized atom that is closed under substructures and polynomial time computable. Then $\sat(\exists^*\forall^* [A_Q])$ and $\finsat(\exists^*\forall^* [A_Q])$ are
in $2\NEXPTIME$. If $\tau$ is a vocabulary consisting of relation symbols of arity at most $k$, $k\in \mathbb{Z}_+$, then $\sat(\exists^*\forall^* [A_Q](\tau))$ and $\finsat(\exists^*\forall^* [A_Q](\tau))$ are 
 $\NEXPTIME$-complete.
 \end{theorem}
\begin{proof} Note first that the lower bounds follow from the fact
that both $\sat(\exists^*\forall^*)$ and $\finsat(\exists^*\forall^*)$
are already  $\NEXPTIME$-complete.
It hence suffices to show containments in $2\NEXPTIME$ and $\NEXPTIME$, respectively.
Let $\varphi \in \exists^*\forall^* [A_Q]$. By Lemma \ref{AE-lemma}, $\varphi$ is satisfiable if and only if  it has a model of cardinality at most $|\varphi |$. We can decide satisfiability of $\varphi$ as follows:  non-deterministically guess a structure $\mA$ of cardinality at most $|\varphi |$
and accept iff $\mA \models \varphi$. By the result of Gr\"{a}del in \cite{DBLP:journals/tcs/Gradel13},
the question whether $\mA \models \varphi$ can be checked non-deterministically in
exponential time with input $\mA$ and  $\varphi$.
Assume first that the maximum arity of relation symbols that may occur in $\varphi$ is not a fixed constant.
Relation symbols of arity at most $\lvert \varphi\rvert$ may occur in $\varphi$.
Thus the size of the \emph{binary encoding of a model}
$\mathfrak{A}$ of $\varphi$ such that $A \leq |\varphi|$
is worst case exponential with
respect to $\lvert \varphi\rvert$.
If, on the other hand, the maximum arity of relation symbols that can occur in $\varphi$ is a fixed constant,
then the size of the encoding of $\mA$ is just worst case polynomial with respect to $\lvert \varphi\rvert$.
Therefore it follows that our algorithm for checking satisfiability of $\varphi$ is in $\NEXPTIME$ in the case of fixed arity vocabularies and in $2\NEXPTIME$ in the general case. The corresponding results for the finite satisfiability problem follow by the observation that $\exists^*\forall^* [A_Q]$ has the finite model property, Lemma \ref{AE-lemma}.
\end{proof}
\begin{corollary} Let $\mathcal{A}$ be a
uniformly polynomial time computable
class of generalized atoms that are closed under substructures.
Then $\sat(\exists^*\forall^* [\mathcal{A}])$ and $\finsat(\exists^*\forall^* [\mathcal{A}])$ are 
 in $2\NEXPTIME$. If $\tau$ is a vocabulary consisting of relation symbols of arity at most $k$, $k\in \mathbb{Z}_+$, then $\sat(\exists^*\forall^* [\mathcal{A}](\tau))$ and $\finsat(\exists^*\forall^* [\mathcal{A}](\tau))$ are 
 $\NEXPTIME$-complete.
 \end{corollary}
In the following sense Theorem \ref{EAcomplete} is optimal: there exists a polynomial time computable generalized atom $A_Q$ such that $\sat(\exists^3\forall[A_Q])$ and $\finsat(\exists^3\forall[A_Q])$ are undecidable. This already holds for vocabularies with at least one binary relation symbol and a countably infinite set of unary relation symbols. Let $\varphi_{\mathit{5-inc}}:= \forall x_1\dots\forall x_5\big(R(x_1,\dots,x_5)\rightarrow S(x_1,\dots x_5) \big)$, and let $A_{\mathit{5-inc}}$ be the related generalized atom of the type $(5,5)$, i.e., $A_{\mathit{5-inc}}$ is the $5$-ary inclusion atom interpreted as a generalized atom. Clearly $A_{\mathit{5-inc}}$ is computable in polynomial time.
%
\begin{theorem}\label{whatevertheorem}
Let $\tau$ be a vocabulary consisting of one binary relation symbol and a countably infinite set of unary relation symbols. Then
both $\sat(\exists^3\forall[A_{\mathit{5-inc}}](\tau))$ and $\finsat(\exists^3\forall[A_{\mathit{5-inc}}](\tau))$ are undecidable.
\end{theorem}
\begin{proof}
It well known that for the Kahr class (i.e., the prefix class $\forall\exists\forall$ of $\mathrm{FO}$ with vocabulary $\tau$) the satisfiability and the finite satisfiability problems are undecidable (see, e.g., \cite{graadeli}). From the proof of \cite[Theorem 5]{Hannula14} it follows that there exists a polynomial time translation $\varphi\mapsto \varphi^*$ from the Kahr class into $\exists^3\forall[A_{\mathit{5-inc}}](\tau)$ such that
\(
\mA\models_X \varphi\Leftrightarrow \mA\models_X \varphi^*
\)
holds for every model $\mA$ and team $X$ with codomain $A$. Thus $\sat(\exists^3\forall[A_{\mathit{5-inc}}](\tau))$ and $\finsat(\exists^3\forall[A_{\mathit{5-inc}}](\tau))$ are undecidable.
\end{proof}

%
It is easy to see that  dependence atoms viewed as generalized atoms are closed under substructures because they are both downwards closed and universe independent. Likewise, it is straightforward to check that the class of dependence atoms is uniformly polynomial time computable. Hence we obtain the following corollary.
\begin{corollary}Both the satisfiability and the finite satisfiability problems for the $\exists^*\forall^*$-sentences of dependence logic are in  $2\NEXPTIME$.  If $\tau$ is a vocabulary consisting of relation symbols of arity at most $k$, then the satisfiability and the finite satisfiability problems for the $\exists^*\forall^*$-sentences of dependence logic over the vocabulary $\tau$ are
$\NEXPTIME$-complete.
\end{corollary}
%


\vspace{-5mm}

\section{Conclusion}

\vspace{-2mm}

%
%

We have tied some loose 
ends concerning the complexity of predicate logics 
based on team semantics. 
Using a general approach, we have shown that  the satisfiability and the finite satisfiability problems of the two-variable fragments of
inclusion logic, exclusion logic, inclusion/exclusion logic, and independence logic are all $\NEXPTIME$-complete. 
Additionally, we have shown that the satisfiability and the finite satisfiability problems of the prefix class $\exists^*\forall^*$ of dependence logic are 
$\NEXPTIME$-complete for any vocabulary of bounded arity, and in $2\NEXPTIME$ in the general case.
The general approach we have employed of
course also implies a range of other results
on team-semantics-based logics.
Finally, we have proved that the validity problem of two-variable dependence logic is  undecidable, thereby answering an open problem from the
literature on team semantics.
%

This article clears path to a more comprehensive classification of the decidability and complexity of different fragments of logics with
generalized atoms and team semantics. In the future, we aim to identify
further interesting related systems with a decidable satisfiability problem.\\

\bibliographystyle{plain}
\bibliography{tampere}

\end{document}